\documentclass[11pt]{article}
\usepackage{nicefrac}
\usepackage{amsfonts}
\usepackage{amssymb}
\usepackage{color}
\usepackage{amsmath,mathrsfs,dsfont}
\usepackage{graphicx}
\usepackage{appendix}
\usepackage{cite}
\newtheorem{theorem}{Theorem}

\newtheorem{definition}[theorem]{Definition}

\newtheorem{proposition}[theorem]{Proposition}

\newenvironment{proof}[1][Proof]{\noindent\textbf{#1.} }{\ \rule{0.5em}{0.5em}}

\makeatletter
\let\@fnsymbol\@arabic
\makeatother

\begin{document}

\title{Heat kernel for flat generalized Laplacians with anisotropic scaling.}
\author{A. Mamiya\thanks{arthurmamiya@gmail.com}\ \  and A. Pinzul\thanks{apinzul@unb.br}\\
\\
\emph{Universidade de Bras\'{\i}lia}\\
\emph{Instituto de F\'{\i}sica}\\
\emph{70910-900, Bras\'{\i}lia, DF, Brasil}\\
\emph{and}\\
\emph{{International Center of Condensed Matter Physics} }\\
\emph{C.P. 04667, Brasilia, DF, Brazil} \\
}
\date{}
\maketitle

\begin{abstract}
We calculate the closed analytic form of the solution of heat kernel equation for the anisotropic generalizations of flat Laplacian. We consider a UV as well as UV/IR interpolating generalizations. In all cases, the result can be expressed in terms of Fox-Wright psi-functions. We perform different consistency checks, analytically reproducing some of the previous numerical or qualitative results, such as spectral dimension flow. Our study should be considered as a first step towards the construction of a heat kernel for curved Ho\v{r}ava-Lifshitz geometries, which is an essential ingredient in the spectral action approach to the construction of the Ho\v{r}ava-Lifshitz gravity.
\end{abstract}

\section{Introduction}
Recently, space-time with anisotropic scaling has become a popular subject in theoretical and mathematical physics. Largely, this is due to the influential paper by Ho\v{r}ava \cite{Horava:2009uw} where a candidate for the possible UV completion of Quantum Gravity was suggested. This work is a non-trivial development of Lifshitz' earlier idea \cite{Lifshitz} (Hence the common name for this type of models - Ho\v{r}ava-Lifshitz (HL) models). Since then the area has attracted considerable attention. The developments have been along several different directions: study of the possible experimental consequences, theoretical study of the structure of the model, different generalizations etc. For some recent reviews, see \cite{Padilla:2010ge,Sotiriou:2010wn,Mukohyama:2010xz,Clifton:2011jh}.

Our interest in HL model lies in the detailed study of this model from the point of view of spectral (or non-commutative geometry) \cite{Connes:1994yd}. The possible advantages of this approach were discussed in \cite{Pinzul:2010ct} where one of the first steps in this direction had been made. Namely, in \cite{Pinzul:2010ct} we calculated the spectral dimension of the HL model of arbitrary curvature using the methods of spectral geometry. The result was in complete agreement with the earlier calculation for the flat case \cite{Horava:2009if}. To perform this analysis, we introduced a notion of a generalized Laplacian, which is an appropriate generalization of the usual Laplacian on a manifold. This modification is needed to incorporate the information about the anisotropic scaling. This generalized Laplacian (or its counterpart in the form of a generalized Dirac operator) is one of the main ingredients in the spectral approach to the HL gravity. This is due to a couple of reasons:

1) It is well known that the information about the usual commutative geometry can be extracted from the so-called spectral triple \cite{Connes:1994yd,Connes:2008vs}. Dirac operator is an essential component of this triple (the other two being the algebra of functions and the Hilbert space of spinors). In the same way, we can construct generalized geometries by using the generalized spectral triples. In the case of the HL models, what is really modified is \textit{physical} geometry, i.e. the geometry as seen by the physical objects (fields etc.), see the discussion in \cite{Pinzul:2010ct,Gregory:2012an}. The generalized Dirac or Laplace operators encode the information about this physical geometry: a) the space-time is a foliated manifold (as in ADM approach); b) the time scales differently from the space.

2) At the same time, the Dirac operator contains the information about the dynamics of both geometry and matter. Essentially, this is the main idea behind the spectral action principle \cite{Chamseddine:1996zu}.
The great advantage of this approach is in the fact that the same object (Dirac operator) controls both geometric and matter part of the action. And we see this as a very promising for applications to HL models where the action contains huge ambiguity. The success of the spectral action principle applied to the Standard Model \cite{Chamseddine:2008zj} is very encouraging.

While in \cite{Pinzul:2010ct} the calculation was motivated by the first type of the reasonings - we calculated some geometrical characteristics of the typical HL space-time (spectral dimension), here we would like to make some initial steps in the other direction. The knowledge of the heat kernel of the Dirac operator is essential for the calculation of the geometric part of the spectral action \cite{Chamseddine:1996zu}. It is reasonable to start with the heat kernel for the flat generalizations of Laplacian (as we will see, already this is not a completely trivial task). The result is not just of some methodological interest. Rather it should serve as a starting point for the analysis of the general curved case, which is currently under study. The final outcome of this line of research should be purely algebraic derivation of the action for the HL gravity via the spectral action principle. This will be the focus of our future research.

The plan of the paper is the following. In Section \ref{1}, we study in details the case of a UV generalization of Laplacian. Section \ref{2} is devoted to thorough analysis of the general flat case. Multiple analytic results are obtained in complete agreement with previous studies as well as with physical intuition. In Conclusions we discuss the importance of the obtained results and outline the future directions. Finally, in Appendix, we briefly review some properties of the Fox-Wright psi-function, providing some proofs of the facts used in the main text.

\section{Anisotropic flat Laplacian. UV case.}\label{1}

The heat kernel technique is quite powerful and has been applied with great success in many different areas of theoretical and mathematical physics as well as in pure mathematics (see, e.g., \cite{Vassilevich:2003xt} for review). Our interest in heat kernel is due to its intimate relation to the spectral action \cite{Chamseddine:1996zu}. Roughly speaking, one can say that knowing the heat kernel of some relevant Dirac (or Laplace) operator is equivalent to knowing what is the geometrical part of the physical action. There are two important points in this statement:

1) What is this physically relevant operator? By now, it is well known that for the standard choice of the Dirac operator one recovers the Einstein-Hilbert-Yang-Mills action \cite{Chamseddine:1996zu}. But what should be our choice if we want to employ this technique for some generalizations, like HL gravity?

2) Even if the relevant operator has been chosen, it is still a non-trivial task to find a heat kernel for this operator. In fact, even in the case of the standard choice, this result can be found only in the form of some asymptotic series (which is enough to produce the spectral action!). The knowledge of the exact results for some specific cases is essential for this calculation \cite{Vassilevich:2003xt}.

Based on these two points, it should be clear that the most natural first step would be to start with some appropriate generalization of the \textit{flat} Laplacian. In \cite{Pinzul:2010ct}, it was argued that such operator for the case of the foliated geometry with anisotropic scaling between space and time should be of the following form\footnote{This will capture UV structure of the space-time, hence the title of this section. More general case, which is also sensitive to IR, will be considered in the next section. See \cite{Pinzul:2010ct} and the discussion below for more details.}
\begin{eqnarray}
\mathfrak{L}=\partial_t^2 + (-1)^{z+1}(\partial_i\partial_i)^z \ .\nonumber
\end{eqnarray}
The index $i$ takes values $1,2,3$ and $\delta_{ij}$ is used to contract the indices. $z$ is the anisotropic scaling exponent, i.e. under
\begin{eqnarray}
& &t \rightarrow a^z t  \nonumber\\
& &\vec{x} \rightarrow a\vec{x}  \nonumber
\end{eqnarray}
$\mathfrak{L}$ has a well-defined scaling: $\mathfrak{L} \rightarrow a^{-2z}\mathfrak{L}$. The $(-1)^{z+1}$ factor in front of $(\partial_i\partial_i)^z$ in (\ref{heatkernel}) is to insure the generalized ellipticity (see the discussion in \cite{Pinzul:2010ct}).

So, we would like to find a solution to the following problem:
\begin{eqnarray}\label{heatkernel}
& &\partial_\tau K(x-x';\tau) - \mathfrak{L}K(x-x';\tau)=0 \\
& &K(x-x';+0)=\delta^{(4)}(x-x') \ .\nonumber
\end{eqnarray}
The formal solution of (\ref{heatkernel}) is given by
\begin{eqnarray}
& &K(x-x';\tau)=\langle x|\mathrm{e}^{\tau\mathfrak{L}} | x'\rangle\ , \ \mathrm{ where} \nonumber\\
& &\langle x| x'\rangle = \delta^{(4)}(x-x')\ .\nonumber
\end{eqnarray}
Here $|x\rangle$ is the standard generalized eigen-vector of $x$, while $p_i := -\mathrm{i} \partial_i$.
Because in the momentum representation
\begin{eqnarray}
& &\langle p|\mathrm{e}^{\tau\mathfrak{L}} | x\rangle = \mathrm{e}^{-\tau(p_0^2+\vec{p}^{2z})}\langle p| x\rangle \ ,\ \ \langle p | x\rangle =\frac{1}{(2\pi)^2}\mathrm{e}^{-\mathrm{i}px}\ , \nonumber
\end{eqnarray}
we have (as usual, $px:=p_0 t+\vec{p}\vec{x}$)
\begin{eqnarray}
& &K(x-x';\tau)= \frac{1}{(2\pi)^4}\int d^4 p\, \mathrm{e}^{\mathrm{i}p(x-x')-\tau (p_0^2 + \vec{p}^{2z})}\ . \nonumber
\end{eqnarray}
Note that the generalized ellipticity condition is crucial for the convergence of this integral. The integral over $p_0$ is trivially taken producing the standard result for the 1-d kernel:
\begin{eqnarray}\label{K0}
& &\frac{1}{2\pi}\int d p_0\, \mathrm{e}^{\mathrm{i}p_0(t-t')-\tau p_0^2}= \frac{1}{\sqrt{4\pi\tau}}\exp\left(-\frac{(t-t')^2}{4\tau}\right)=: K_0(t-t';\tau)\ .
\end{eqnarray}
Using the $SO(3)$ symmetry, the expression for $K(x-x';\tau)$ can be written in the following form
\begin{eqnarray}\label{K1}
& &K(x-x';\tau)= \frac{1}{2\pi^2 |\vec{x}-\vec{x}'|}K_0(t-t';\tau) \int\limits^\infty_0 d p\, p\, \sin(p|\vec{x}-\vec{x}'|)\, \mathrm{e}^{-\tau p^{2z}}\ .
\end{eqnarray}
We need to calculate the following integral
\begin{eqnarray}\label{I}
& &I = \int\limits^\infty_0 d p\, p\, \sin(\alpha p)\, \mathrm{e}^{-\beta p^{2z}}\ .
\end{eqnarray}
This is done by expanding $\sin(\alpha p)$
\begin{eqnarray}\label{expansion}
I & = & \sum_{k=0}^\infty\frac{(-1)^k}{(2k+1)!}\alpha^{2k+1}\int\limits^\infty_0 d p\, p^{2k+2}\, \mathrm{e}^{-\beta p^{2z}}= \nonumber\\
& = & \frac{1}{2z\alpha^2}\sum_{k=0}^\infty\frac{(-1)^k}{\Gamma(2k+2)}\left(\frac{\alpha}{\beta^{1/2z}}\right)^{2k+1} \int\limits^\infty_0 d p\, p^{\frac{2k+3}{2z}-1}\, \mathrm{e}^{-p} = \nonumber\\
& = & \frac{1}{2z\alpha^2}\left(\frac{\alpha}{\beta^{1/2z}}\right)^3 \sum_{k=0}^\infty \left(-\frac{\alpha^2}{\beta^{1/z}}\right)^{k}\frac{\Gamma(\frac{k+3/2}{z})}{\Gamma(2k+2)} \ .
\end{eqnarray}
Using the Gauss' multiplication formula for gamma-function, we can write for $\Gamma(2k+2)$
\begin{eqnarray}
\Gamma(2k+2)=\frac{2^{2k+3/2}}{\sqrt{2\pi}}\Gamma(k+1)\Gamma(k+3/2) \ .\nonumber
\end{eqnarray}
This allows us to re-write (\ref{expansion}) in the following form
\begin{eqnarray}\label{final1}
& &I = \frac{1}{2z\alpha^2}\left(\frac{\alpha}{\beta^{1/2z}}\right)^3 \frac{\sqrt{\pi}}{2}\sum_{k=0}^\infty \left(-\frac{\alpha^2}{4\beta^{1/z}}\right)^{k}\frac{1}{k!}\frac{\Gamma(\frac{k+3/2}{z})}{\Gamma(k+3/2)}\ .
\end{eqnarray}
The sum in (\ref{final1}) is immediately recognized as (1,1) Fox-Wright psi-function, ${}_1\!\Psi_1$ \cite{Wright:paper,Srivastava:book}, so the final answer for (\ref{I}) takes the form\footnote{It is interesting that the simple looking integral (\ref{I}) leads to this non-trivial result in terms of a relatively rare special function. As a curious observation, we note that this seems to be absent even in \cite{Gradshteyn}.}
\begin{eqnarray}\label{final2}
& &I = \frac{1}{2z\alpha^2}\left(\frac{\alpha}{\beta^{1/2z}}\right)^3 \frac{\sqrt{\pi}}{2} {}_1\!\Psi_1\left[(3/2z,1/z);(3/2,1);-\frac{\alpha^2}{4\beta^{1/z}}\right]\ .
\end{eqnarray}
In appendix, we give the definition of Fox-Wright psi-function and review some of its properties relevant for us.

Now, we use the result (\ref{final2}) to obtain a closed expression for our heat kernel. Recall that $\alpha = |\vec{x}-\vec{x}'|$ and $\beta = \tau$. Then using (\ref{final2}) and (\ref{K0}) in (\ref{K1}) we get
\begin{eqnarray}\label{final}
& &K(x-x';\tau) = \frac{1}{z (4\pi)^2\tau^{\frac{1}{2}(1+3/z)}}\, \mathrm{e}^{-\frac{(t-t')^2}{4\tau}} {}_1\!\Psi_1\left[(3/2z,1/z);(3/2,1);-\frac{|\vec{x}-\vec{x}'|^2}{4\tau^{1/z}}\right]\ .
\end{eqnarray}

Let us do some preliminary analysis of this result.

1) $\textbf{z=1}$. Using that ${}_1\!\Psi_1\left[(a,A);(a,A); x\right]= \exp(x)$, we immediately recover the classical result for the standard Laplacian:
\begin{eqnarray}
& &K(x-x';\tau) = \frac{1}{(4\pi\tau)^2}\, \mathrm{e}^{-\frac{(x-x')^2}{4\tau}} \ .\nonumber
\end{eqnarray}

2) \textbf{The limit $x \rightarrow x'$}. In this limit for the arbitrary $z$ the heat kernel (\ref{final}) is reduced to
\begin{eqnarray}
K(0;\tau) = \frac{\Gamma(3/2z)}{z\Gamma(3/2)}\,\frac{1}{(4\pi)^2\tau^{\frac{1}{2}(1+3/z)}}\ .\nonumber
\end{eqnarray}
It is well-known (see, e.g. \cite{Ambjorn:2005db}) that this is related to the return probability for a random walk on our space-time, which, in its turn, is related to the so-called spectral dimension $d_S$ by (also see the discussion in next section)
\begin{eqnarray}\label{spectral1}
K(0;\tau) \sim \frac{1}{\tau^{\frac{d_S}{2}}}\ .
\end{eqnarray}
So we see that we (trivially) reproduce the known result for the Horava-Lifshitz space-time \cite{Pinzul:2010ct,Horava:2009if}:
\begin{eqnarray}\label{HLd}
d_S = 1+\frac{3}{z}\ .
\end{eqnarray}

3) \textbf{Other representations of $K(x-x';\tau)$}. Here, for the convenience of the possible future applications, we give other forms of the heat kernel (\ref{final}) in terms of hypergeometric functions and Meijer's $G$-function. This is done with the help of the general formulas from the appendix, Eqs.(\ref{Hyper}),(\ref{Meijer}). Noting that in our case
\begin{eqnarray}
a=\frac{3}{2z}\, ,\ \alpha =1\, ,\ b=\frac{3}{2}\, ,\ \beta = z\ \mathrm{and}\ k=z  \nonumber
\end{eqnarray}
(all the notations as in the appendix), we have

i) Hypergeometric representation:
\begin{align}\label{KHyper}
&K(x-x';\tau) = \frac{1}{z (4\pi)^2\tau^{\frac{1}{2}(1+3/z)}}\, \mathrm{e}^{-\frac{(t-t')^2}{4\tau}} \left( \frac{\Gamma(\frac{3}{2z})}{\Gamma(\frac{3}{2})} + \sum_{s=1}^z \frac{\Gamma(\frac{3+2s}{2z})}{\Gamma(\frac{3+2s}{2})}{\left( -\frac{|\vec{x}-\vec{x}'|^2}{4\tau^{1/z}} \right)}^s\frac{1}{s!}\,\times \right.\nonumber\\
& \left. \times{}_1\! F_{2z-1}\left[ 1; \frac{2s+5}{2z},..,\frac{2s+3+2(z-1)}{2z},\frac{s+1}{z},...,\frac{s+z}{z} ; (-1)^z\frac{|\vec{x}-\vec{x}'|^{2z}}{(2z)^{2z}\tau} \right]\right) \ .
\end{align}
To arrive at this result, we have used the trivial parameter cancellation property for the hypergeometric functions:
$$
{}_n F_m\left[ ...,a,...; ..,a,...; x \right] = {}_{n-1} F_{m-1}\left[ ...; ...; x \right] \ .
$$

ii) Meijer's $G$-function representation:
\begin{align}\label{KMeijer}
&K(x-x';\tau) = \frac{2}{(4\pi)^2\tau^{\frac{1}{2}(1+3/z)}}\left(\frac{\pi}{z}\right)^z\, \mathrm{e}^{-\frac{(t-t')^2}{4\tau}}\times\nonumber\\
&\times G^{z\,1}_{1\,2z}\left( \frac{|\vec{x}-\vec{x}'|^{2z}}{(2z)^{2z}\tau}  \left|
         \begin{array}{c}
           1-\frac{3}{2z} \\
           0,\frac{1}{z},...,\frac{z-1}{z},\frac{2z-3}{2z},...,\frac{2z-(3+2(z-1))}{2z} \\
         \end{array}
 \right.\right)\ .
\end{align}

The analytic result for the heat kernel in either of the three forms, (\ref{final}), (\ref{KHyper}) or (\ref{KMeijer}), does not look simple. This shows that one should expect that the analysis of the general curved case would be considerably more involved than for the standard Laplacian.

\section{Anisotropic flat Laplacian. UV/IR interpolating case.}\label{2}
The results of the previous section can be readily extended to operators of the form  $\mathfrak{L}=\partial_t^2 + \sum_{k=0}^z (-1)^{k+1} \gamma_k(\partial_i\partial_i)^k$, $k \le z$, where $\gamma_k$ are some constants.\label{footnote}\footnote{Actually, $\gamma$'s are related to some characteristic scales $M_k$, $\gamma_k \sim M_k^{2(z-k)}$ \cite{Pinzul:2010ct}.} One immediately obtains the generalization of (\ref{K1}):
\begin{eqnarray}
& &K(x-x';\tau)= \frac{1}{2\pi^2 |\vec{x}-\vec{x}'|}K_0(t-t';\tau) \int\limits^\infty_0 d p\, p\, \sin(p|\vec{x}-\vec{x}'|)\, \mathrm{e}^{-\tau \sum\limits_{k=0}^z  \gamma_k p^{2^k}}\ .\nonumber
\end{eqnarray}
So, the analog of (\ref{I}) that needs to be calculated is
\begin{eqnarray}\label{Kgeneral}
& & I = \int\limits^\infty_0 d p\, p\, \sin(\alpha p)\, \mathrm{e}^{-\sum\limits_{k=0}^z  \beta_k p^{2k}}\ . \nonumber
\end{eqnarray}
The strategy is the same as was used to evaluate (\ref{I}):
\begin{eqnarray}\label{K3}
& & I =  \int\limits^\infty_0 d p\, p\, \mathrm{e}^{-\beta_z p^{2z}}\left[  \prod_{k=0}^{z-1}\sum_{j_k=0}^\infty \frac{1}{{j_k}!} (-\beta_k p^{2k})^{j_k} \right] \, \left[ \sum_{l=0}^\infty \frac{(-1)^l}{(2l+1)!} (\alpha p)^{2l+1} \right] = \\
& = &  \frac{1}{2\pi^2 \alpha} \sum_{\{j_k\} = 0}^{\infty} \sum_{l=0}^{\infty} \frac{(-1)^l}{(2l+1)!}\alpha^{2l+1} \left(\prod_{k=0}^{z-1}{\frac{(-\beta_k)^{j_k}}{j_k!}}\right)\int_0^\infty p^{(2+2l+\sum_k k j_k)/z} e^{-\beta_z p^{2z}} dp  \nonumber\\
& = & \frac{1}{4\pi^2 z} \sum_{\{j_k\} = 0}^{\infty} \sum_{l=0}^{\infty}\left( \prod_{k=0}^{z-1}{\frac{(-\beta_k)^{j_k}}{j_k!}}\right) \frac{\beta_z^{-(3/2+\sum_k k j_k)/z}}{(2l+1)!}(-\alpha^2 \beta_z^{-1/z})^{l} \Gamma\left(\frac{3/2 + l + \sum_k k j_k}{z}\right) \nonumber\\
& = & \frac{\sqrt{\pi}}{8\pi^2 z} \sum_{\{j_k\} = 0}^{\infty}\left( \prod_{k=0}^{z-1}{\frac{(-\beta_k)^{j_k}}{j_k!}}\right)\beta_z^{-(3/2+\sum_k k j_k)/z}{}_1\!\Psi_1 \left[ \left(\frac{3/2 +\sum_k k j_k}{z}, 1/z\right);(3/2,1);-\frac{\alpha^2}{4\beta_z^{1/z}} \right] \ ,\nonumber
\end{eqnarray}
where in the last step we again used the Gauss' multiplication formula. This leads to the heat kernel in the form
\begin{eqnarray}\label{heatkernel2}
& K(x-x';\tau) = \frac{1}{z (4\pi)^2 \tau^{\frac{1}{2}{(1+3/z)}}} \mathrm{e}^{-\frac{(t-t')^2}{4\tau}} \sum\limits_{\{j_k\} = 0}^{\infty}\left( \prod\limits_{k=0}^{z-1}{\frac{(-\tau \gamma_k)^{j_k}}{j_k!}}\right)(\tau \gamma_z)^{-\sum_k k j_k /z} \, \times\nonumber\\
&\times  {}_1\!\Psi_1 \left[ ((3/2 +\sum_k k j_k)/z, 1/z);(3/2,1);-\frac{|\vec{x}-\vec{x}'|^2}{4(\tau \gamma_z)^{1/z}} \right]\ .
\end{eqnarray}
As the first (trivial) check of this result, we should reproduce our heat kernel from the previous section, Eq.(\ref{final}). This is easily achieved by setting $\gamma_k = 0$ for all $ k \neq z$ and $\gamma_z = 1$.

Our further analysis we will perform for the slightly less general case when only one additional $\gamma_k$ is not zero, i.e. for the generalized Laplacian of the form
\begin{eqnarray}\label{simplifiedLap}
\mathfrak{L}=\partial_t^2 + (-1)^{z+1}\gamma_z(\partial_i\partial_i)^z +(-1)^{k+1} \gamma(\partial_i\partial_i)^k \ .
\end{eqnarray}
This will allow us to avoid some technical complications and yet preserve all the interesting physical features as "dimensional transmutation" and so on. For this specific case, the general formula (\ref{heatkernel2}) will slightly simplify
\begin{eqnarray}\label{heatkernelsimple}
& K(x-x';\tau) = \frac{1}{z (4\pi)^2 \tau^{\frac{1}{2}{(1+3/z)}}} \mathrm{e}^{-\frac{(t-t')^2}{4\tau}}   \sum\limits_{j=0}^\infty \frac{1}{j !} {\left[ -\gamma {\tau}^{1-k/z} \right]}^{j} \, \times\nonumber\\
&\times  {}_1\!\Psi_1\left[((3/2+k j)/z,1/z);(3/2,1);-\frac{|\vec{x}-\vec{x}'|^2}{4{\tau}^{1/z}} \right] \ .
\end{eqnarray}
We would like to address two questions:

1) What is the spectral dimension of the model based on (\ref{simplifiedLap})?

2) Can we show analytically that in IR (the precise meaning of this will be given later) the system is indistinguishable from the one with the generalized Laplacian $\mathfrak{L}=\partial_t^2 + (-1)^{k+1} \gamma(\partial_i\partial_i)^k$?

To answer both of these questions, another representation of our result (\ref{heatkernel2}) will prove useful. We will work it out for our specific case. For this case the penultimate step in (\ref{K3}) will look like
\begin{eqnarray}
I = \frac{1}{4\pi^2 z} \sum_{j = 0}^{\infty} \sum_{l=0}^{\infty}\frac{(-\beta)^{j}}{j!} \frac{\beta_z^{-(3/2+k j)/z}}{(2l+1)!}(-\alpha^2 \beta_z^{-1/z})^{l} \Gamma\left(\frac{3/2 + l + k j}{z}\right) \ .\nonumber
\end{eqnarray}
Now we have two options: either take the sum over $l$ or over $j$. In the first case we will arrive at the result (\ref{heatkernel2}) in terms of ${}_1\!\Psi_1$. The second one will lead to a new representation in terms of ${}_1\!\Psi_0$:\footnote{I.e. relation between (\ref{K3}) and (\ref{I3}) below has the form of a re-summation formula.}
\begin{eqnarray}\label{I3}
I = \frac{1}{4\pi^2 z} \sum_{l = 0}^{\infty} \frac{\beta_z^{-3/2z}}{(2l+1)!}(-\alpha^2 \beta_z^{-1/z})^l {}_1\!\Psi_0\left[\left( \frac{3+2l}{2z},\frac{k}{z}\right);-\beta\beta_z^{-k/z}\right] \ .
\end{eqnarray}
Using (\ref{I3}), we have an equivalent representation for the heat kernel:
\begin{eqnarray}\label{heatkernel3}
& K(x-x';\tau)= \frac{2}{z \gamma_z^\frac{3}{2z} \sqrt{\pi}(4\pi)^2 \tau^{\frac{1}{2}{(1+\frac{3}{z})}}} \mathrm{e}^{-\frac{(t-t')^2}{4\tau}} \sum\limits_{l = 0}^{\infty} \frac{(-|\vec{x}-\vec{x}'|^2(\tau \gamma_z)^{-\frac{1}{z}})^l}{(2l+1)!} \times \nonumber \\
&\times {}_1\!\Psi_0\left[\left( \frac{3+2l}{2z},\frac{k}{z}\right);-\gamma\gamma_z^{-\frac{k}{z}}\tau^{1-\frac{k}{z}}\right] \ .
\end{eqnarray}

\subsection{Spectral dimension}

Now we would like to analyze the spectral dimension, which follows from (\ref{heatkernel3}). We recall that the spectral dimension can be read from the return probability $K(0;\tau)$, see the discussion around Eq.(\ref{spectral1}). Using (\ref{spectral1}) we can \textit{define} the spectral dimension of generalized geometries as \cite{Ambjorn:2005db,Horava:2009if}
\begin{eqnarray}
d_S = -2\frac{d\ln K(0;\tau)}{d\ln\tau}\ .\nonumber
\end{eqnarray}
Applying this to our case with $K(0;\tau)$ given by (only $l=0$ term will contribute in (\ref{heatkernel3}) when $|x-x'|=0$)
\begin{eqnarray}\label{heatkernel0}
& K(0;\tau)= \frac{2}{z \gamma_z^\frac{3}{2z} \sqrt{\pi}(4\pi)^2 \tau^{\frac{1}{2}{(1+\frac{3}{z})}}} {}_1\!\Psi_0\left[\left( \frac{3}{2z},\frac{k}{z}\right);-\gamma\gamma_z^{-\frac{k}{z}}\tau^{1-\frac{k}{z}}\right] \ ,\nonumber
\end{eqnarray}
we get the analytic answer for the \textit{running} spectral dimension
\begin{eqnarray}\label{spectralfinal}
d_S = 1+\frac{3}{z} + 2\gamma\gamma_z^{-\frac{k}{z}}\tau^{1-\frac{k}{z}}\left(1-\frac{k}{z}\right) \frac{{}_1\!\Psi_0\left[\left( \frac{3+2k}{2z}, \frac{k}{z}\right);-\gamma\gamma_z^{-\frac{k}{z}}\tau^{1-\frac{k}{z}}\right]}{{}_1\!\Psi_0\left[\left( \frac{3}{2z},\frac{k}{z}\right);-\gamma\gamma_z^{-\frac{k}{z}}\tau^{1-\frac{k}{z}}\right]} \ .
\end{eqnarray}
In derivation of this result we used the Proposition \ref{derivative} from the appendix on derivative of psi-function.

To analyze (\ref{spectralfinal}), we need to clarify what is understood under UV or IR regime. What we expect physically is that in both limits the spectral dimension will be given by the formula (\ref{HLd}), with the only difference that in IR instead of $z$ the dimension should be controlled by $k$. We will show below that this can be demonstrated analytically. So, what plays the role of the parameter controlling the UV/IR transition? Taking into account the footnote on the page \pageref{footnote}, one can easily see that such dimensionless parameter is $\rho:=\gamma\gamma_z^{-\frac{k}{z}}\tau^{1-\frac{k}{z}}$. When $\rho\ll 1$ we are in deep UV and $\rho\gg 1$ corresponds to IR. This could be understood recalling the interpretation of $\tau$ as the return time (not the time of the space-time!). So, $\tau$, small compared to the scale determined by $\gamma$'s, corresponds to probing space-time distances shorter then this scale and vice versa. Now we can analyze (\ref{spectralfinal}) in both regimes.

a) \textbf{UV.} Because ${}_1\!\Psi_0 [(a,A),z]\xrightarrow[|z|\rightarrow 0]{}\Gamma (a)$, it is clear that the second term in (\ref{spectralfinal}) goes to zero in this limit, which gives us the expected result:
\begin{eqnarray}
d^{UV}_S = 1+\frac{3}{z} \ . \nonumber
\end{eqnarray}

b) \textbf{IR.} To analyze this regime, we will use the result of the proposition \ref{asymptotics} of the appendix:
\begin{eqnarray}
{}_1\!\Psi_0 [(a,A),-z]\xrightarrow[|z|\rightarrow\infty]{}\frac{\Gamma (a/A)}{A}z^{-a/A} \ . \nonumber
\end{eqnarray}
Applying this to (\ref{spectralfinal}), one immediately arrives at the IR dimension
\begin{eqnarray}
d^{IR}_S = 1+\frac{3}{k} \ .\nonumber
\end{eqnarray}

So, we can conclude that (\ref{spectralfinal}) provides the analytic answer for the (classical) flow of the spectral dimension, see Fig.(\ref{fig:scaling}) for the case $k=1$. Comparing this to the results of \cite{Pinzul:2010ct} (where the analytic answer for the intermediate regime was not known), one sees perfect agreement.
\begin{figure}[htb]
\begin{center}
\leavevmode
\includegraphics[scale=0.7]{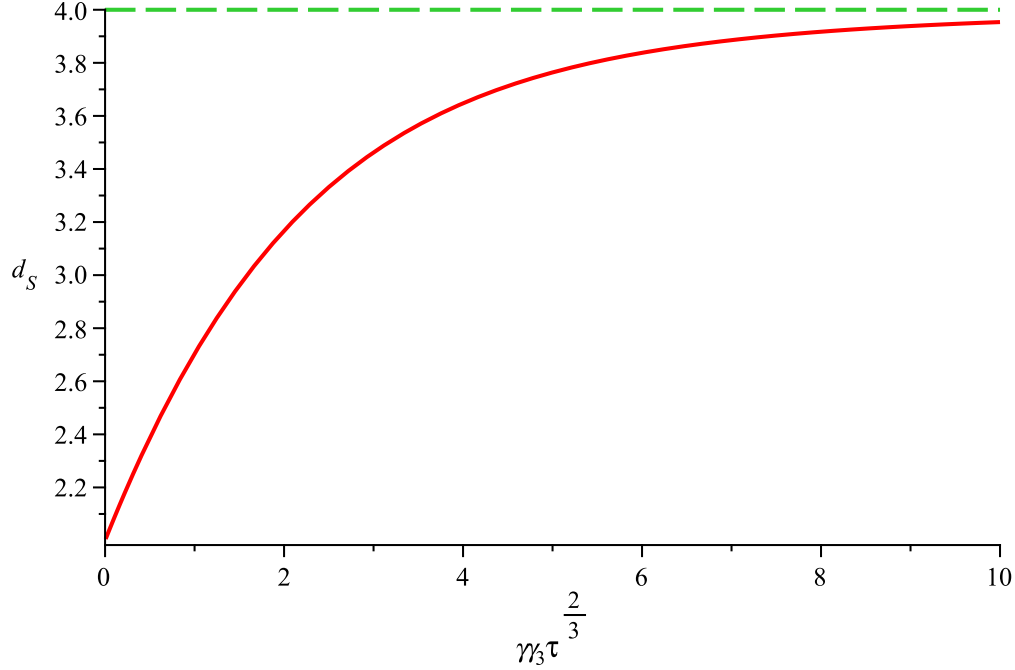}
\end{center}
\caption{The example of a smooth transition from the UV to IR regime for $z=3$ and $k=1$.}
\label{fig:scaling}
\end{figure}

\subsection{UV/IR behaviour}

In this part, we would like to give an analytic proof that the whole dynamics (not just spectral dimension as above) of the system defined by the generalized Laplacian (\ref{simplifiedLap}) in UV is controlled by $\mathfrak{L}^{UV}=\partial_t^2 + (-1)^{z+1}\gamma_z(\partial_i\partial_i)^z$ while in IR by $\mathfrak{L}^{IR}=\partial_t^2 + (-1)^{k+1} \gamma(\partial_i\partial_i)^k$. For this we will show that in both regimes the exact heat kernel (\ref{heatkernelsimple}) reduces to the form (\ref{final}) with the obvious change $z\rightarrow k$ in IR. To do so, we will start with the heat kernel in the form (\ref{heatkernel3}).

a) \textbf{UV.} As we discussed, this corresponds to $\gamma\gamma_z^{-\frac{k}{z}}\tau^{1-\frac{k}{z}} \ll 1$, so we can use
$$
{}_1\!\Psi_0 \left[\left(\frac{3+2l}{2z},\frac{k}{z}\right), -\gamma\gamma_z^{-\frac{k}{z}}\tau^{1-\frac{k}{z}}\right]\rightarrow\Gamma \left(\frac{3+2l}{2z}\right) \ .
$$
Using this in (\ref{heatkernel3}), we have
\begin{eqnarray}\label{hkuv}
& &K^{UV}(x-x';\tau)= \frac{2}{z \gamma_z^\frac{3}{2z} \sqrt{\pi}(4\pi)^2 \tau^{\frac{1}{2}{(1+\frac{3}{z})}}} \mathrm{e}^{-\frac{(t-t')^2}{4\tau}} \sum\limits_{l = 0}^{\infty} \frac{(-|\vec{x}-\vec{x}'|^2(\tau \gamma_z)^{-\frac{1}{z}})^l}{(2l+1)!}\Gamma \left(\frac{3+2l}{2z}\right) = \nonumber\\
& &= \frac{1}{z \gamma_z^\frac{3}{2z}(4\pi)^2\tau^{\frac{1}{2}(1+3/z)}}\, \mathrm{e}^{-\frac{(t-t')^2}{4\tau}} {}_1\!\Psi_1\left[(3/2z,1/z);(3/2,1);-\frac{|\vec{x}-\vec{x}'|^2}{4(\gamma_z \tau)^{1/z}}\right]\ ,
\end{eqnarray}
where we used the definition of ${}_1\!\Psi_1$ from the appendix. It is clear that (\ref{hkuv}) is exactly the heat kernel (\ref{final}).

b) \textbf{IR.} As in the case of the spectral dimension, now we have $\gamma\gamma_z^{-\frac{k}{z}}\tau^{1-\frac{k}{z}} \gg 1$, so we can use asymptotics
$$
{}_1\!\Psi_0 \left[\left(\frac{3+2l}{2z},\frac{k}{z}\right), -\gamma\gamma_z^{-\frac{k}{z}}\tau^{1-\frac{k}{z}}\right]\rightarrow\frac{z}{k}\Gamma \left(\frac{3+2l}{2k}\right) \left(\gamma\gamma_z^{-\frac{k}{z}}\tau^{1-\frac{k}{z}}\right)^{-\frac{3+2l}{2k}}\ .
$$
Substituting this in (\ref{heatkernel3}) and again using the definition of ${}_1\!\Psi_1$, one immediately arrives at the correct IR answer:
\begin{eqnarray}\label{hkir}
& &K^{IR}(x-x';\tau)= \frac{1}{k \gamma^\frac{3}{2k}(4\pi)^2\tau^{\frac{1}{2}(1+3/k)}}\, \mathrm{e}^{-\frac{(t-t')^2}{4\tau}} {}_1\!\Psi_1\left[(3/2k,1/k);(3/2,1);-\frac{|\vec{x}-\vec{x}'|^2}{4(\gamma \tau)^{1/k}}\right]\ . \nonumber
\end{eqnarray}

\section{Conclusions}

In this work we have studied heat kernels for a class of the generalized flat Laplacians. This problem is of interest both mathematically and physically.

Mathematically, in addition to that it is always interesting to find some new solutions of the heat kernel equation, our results of a great interest due to the their connection to generalized geometries. In our particular case, the objects we have studied are essential for the spectral geometrical study of foliated manifolds with anisotropic scaling. Using our analytic results, we were able to explicitly study some interesting geometrical properties of this geometry, such as spectral dimension.

Physically, the heat kernel is a crucial element of the spectral action approach. Namely, it is responsible for the geometrical part of the action, which is given by the trace of the relevant Laplace operator. Of course, in the flat case this part will be trivially zero even for the generalized Laplacian, but as we explained in Introduction, the study of the flat heat kernel is a very important step towards the general case. Once again, having analytic solution allows the detailed analysis. In particular, we have been able to explicitly study the UV/IR transition in heat kernel.

As a bonus outcome of our study, we should mention the active use of Fox-Wright psi-functions. Along with obtaining (or, probably re-deriving?) some properties of these functions, such as mentioned in Section \ref{2} re-summation formula or expression for the integral (\ref{I}) in terms of a psi-function, the main point, in our opinion, is that we bring to the attention of broader scientific community such a powerful yet less known object as the Fox-Wright psi-function. Being more general then hypergeometric functions, it has potential to find its way into many physical problems. This, taking into account the well-studied properties of Fox-Wright psi-functions, might help finding analytic solutions that have been missed before.

The future directions quite naturally follow from the main goal of this line of the research: application of the spectral action principle to HL gravity. The next step should be the study of the heat kernel equation for some less trivial (non-flat) but still not general geometry. The complexity of our flat solutions shows that already this step will be technically a serious challenge and our preliminary study absolutely proves this. Having the solution to this problem, one should try to reproduce HL action as a spectral action for this specific geometry. Ultimately, this should lead to the complete spectral action. This, as we argued, might help to tame the huge ambiguity in the present formulation of HL gravity as well as provide deeper understanding of the model.

\section*{Acknowledgements}

A.P. acknowledges partial support of CNPq under grant no.306068/2012-5. A.M. is supported by CNPq master's scholarship.

\section*{Appendix}

Here we summarize some relevant properties of the Fox-Wright psi-function. Due to the importance for our results, we give proofs to some statements below. Also, we find it useful to do so because the psi-function is somewhat unknown to a non-specialist in the area,\footnote{E.g., we were unable to find this function neither in Maple nor in Mathematica.} so providing proofs would help to get acquainted with this function to some extent. For more details, see \cite{Srivastava:book,Miller:paper,Kilbas:book}.

The Fox-Wright psi-function is defined as a natural generalization of the generalized hypergeometric function:
\begin{definition}
\begin{eqnarray}\label{psifunction}
{}_p\Psi_q\left[(a_1,A_1),...,(a_p,A_p);(b_1,B_1),...,(b_q,B_q); z\right] = \sum_{k=0}^\infty\frac{\prod_{n=1}^p \Gamma(a_n+A_n k)}{\prod_{n=1}^q \Gamma(b_n+B_n k)}\frac{z^k}{k!}\ .
\end{eqnarray}
\end{definition}
It is immediately clear that
\begin{eqnarray}\label{PsiF}
{}_p\Psi_q\left[(a_1,1),...,(a_p,1);(b_1,1),...,(b_q,1); z\right] = \frac{\prod_{n=1}^p \Gamma(a_n)}{\prod_{n=1}^q \Gamma(b_n)}\,{}_p F_q\left[a_1,...,a_p;b_1,...,b_q; z\right]\ . \nonumber
\end{eqnarray}

To study the properties of the Fox-Wright psi-function (such as its asymptotic behaviour) and its relation to some other special functions, it is useful to represent psi-function as a special case of the Fox \textit{H}-function, which is defined as \cite{Srivastava:book}
\begin{definition}
\begin{eqnarray}\label{Hfunction}
 H^{m, \,n}_{p,\,q}\left[ z \left|
         \begin{array}{c}
           (a_1,A_1),...,(a_p,A_p) \\
           (b_1,B_1),...,(b_p,B_p)\\
         \end{array}
 \right.\right] = \frac{1}{2 \pi i } \int_{C} \frac{\prod^m_{j=1}\Gamma(b_j-B_j s)\prod^n_{i=1}\Gamma(1-a_i+A_i s)}{\prod^q_{j=m+1}\Gamma(1-b_j+B_j s)\prod^p_{i=n+1}\Gamma(a_i-A_i s)} z^{s} ds \ ,
\end{eqnarray}
where $C$ is taken to be a contour that separates the poles of $\Gamma(b_j-B_j s)$  and of $\Gamma(1-a_i+A_i s)$, and $a_i,\ b_i$ are such that there are no coinciding poles while $A_i,\ B_i \in \mathds{R}_+$.
\end{definition}
So, the Fox \textit{H}-function itself is a generalization of the Meijer's \textit{G}-function, hence the relation of the Fox-Wright psi-function with \textit{G}-function (see below, Eq.(\ref{Meijer})).
\begin{proposition}\label{psiHrelation} Fox-Wright psi-function is given in terms of the Fox \textit{H}-function as\\
${}_p\Psi_q\left[(a_1,A_1),...,(a_p,A_p);(b_1,B_1),...,(b_q,B_q); z\right]=H^{1, \,p}_{p,\,q+1}\left[-z \left|
         \begin{array}{c}
           (1-a_1,A_1),...,(1-a_p,A_p) \\
           (0,1),(1-b_1,B_1),...,(b_p,B_p)\\
         \end{array}\right.\right] $
\end{proposition}
\begin{proof}
The proof is really straightforward. From the definition (\ref{Hfunction})
$$
H^{1, \,p}_{p,\,q+1}\left[-z \left|
         \begin{array}{c}
           (1-a_1,A_1),...,(1-a_p,A_p) \\
           (0,1),(1-b_1,B_1),...,(b_p,B_p)\\
         \end{array}\right.\right] = \frac{1}{2 \pi i } \int_{C} \frac{\Gamma(-s)\prod^p_{i=1}\Gamma(a_i+A_i s)}{\prod^q_{j=1}\Gamma(b_j+B_j s)} (-z)^{s} ds \ .
$$
$\Gamma(-s)$ has all the poles at $s=n=0,1,2,...$, while $\Gamma(a_i+A_i s)$ has poles at $s=-\frac{n+a_i}{A_i}$. For $a_i, A_i > 0$ these poles are separated by the Mellin-type contour $C=\{\mathrm{Im}\ s = -\epsilon\}$, for some small enough $\epsilon>0$. Closing this contour at infinity in the right half-plane, only poles of $\Gamma(-s)$ will contribute producing
$$
H^{1, \,p}_{p,\,q+1}\left[-z \left|
         \begin{array}{c}
           (1-a_1,A_1),...,(1-a_p,A_p) \\
           (0,1),(1-b_1,B_1),...,(b_p,B_p)\\
         \end{array}\right.\right] = \sum_{k=0}^{\infty}\frac{\prod_{n=1}^p \Gamma(a_n+A_n k)}{\prod_{n=1}^q \Gamma(b_n+B_n k)}\frac{z^k}{k!}\ ,
$$
which is exactly the definition (\ref{psifunction}).
\end{proof}

Essentially using this result, in \cite{Miller:paper} several nice properties of ${}_1\!\Psi_1\left[(a,A);(b,B);z\right]$ were obtained for the case when $A,\ B$ are either integer or rational numbers. The case of the rational parameters is the one relevant for us (see Eq.(\ref{final2})).
One has the following relations between ${}_1\!\Psi_1\left[(a,\frac{\alpha}{k});(b,\frac{\beta}{k}); z\right]$ $\alpha,\beta,k\in \mathds{R}$ and some other special functions (see, e.g. \cite{Gradshteyn} for the definitions and properties of the relevant functions):

i) \textit{Hypergeometric}.
\begin{eqnarray}\label{Hyper}
& &{}_1\!\Psi_1\left[\left(a,\frac{\alpha}{k}\right);\left(b,\frac{\beta}{k}\right); z\right] = \frac{\Gamma(a)}{\Gamma(b)}+ \sum_{s=1}^k \frac{\Gamma(a+\frac{\alpha}{k}s)}{\Gamma(b+\frac{\beta}{k}s)} \frac{z^s}{s!}\times \\
& &\times {}_{\alpha + 1} F_{\beta + k}\left[ 1,\frac{s}{k}+\frac{a}{\alpha},..., \frac{s}{k}+\frac{a+\alpha - 1}{\alpha} ;  \frac{s}{k}+\frac{b}{\beta},...,\frac{s}{k}+\frac{b+\beta-1}{\beta}, \frac{s+1}{k},...,\frac{s+k}{k}; \frac{\alpha^\alpha}{\beta^\beta}\left(\frac{z}{k}\right)^k\right]\ .\nonumber
\end{eqnarray}

ii) \textit{Meijer's $G$-function}.
\begin{eqnarray}\label{Meijer}
& &{}_1\!\Psi_1\left[(a,\frac{\alpha}{k});(b,\frac{\beta}{k}); z\right] = 2\pi^{(1+\beta-\alpha+k)/2}\sqrt{\frac{k\beta}{\alpha}} \frac{\alpha^\alpha}{\beta^\beta}\times \\
& &\times G^{k\,\alpha}_{\alpha\,k+\beta}\left( \frac{\alpha^\alpha}{\beta^\beta}\left(-\frac{z}{k}\right)^k  \left|
         \begin{array}{c}
           1-\frac{a}{\alpha},...,1-\frac{a+\alpha-1}{\alpha} \\
           0,\frac{1}{k},...,\frac{k-1}{k},1-\frac{b}{\beta},...,1-\frac{b+\beta-1}{\beta} \\
         \end{array}
 \right.\right)\ .\nonumber
\end{eqnarray}

Also, using Proposition \ref{psiHrelation}, one can easily study the asymptotic expansion for the large argument of the Fox-Wright psi-function. We will do this for the case relevant for us, i.e. for ${}_1\!\Psi_0\left[(a,A);z\right]$.
\begin{proposition}\label{asymptotics}
The asymptotics of ${}_1\!\Psi_0\left[(a,A);z\right]$ is given by
$$
{}_1\!\Psi_0\left[(a,A);-z\right]\xrightarrow[|z|\rightarrow\infty]{}\frac{\Gamma\left(\frac{a}{A}\right)}{A} z^{-\frac{a}{A}} + \mathcal{O}\left(z^{-\frac{a+1}{A}}\right)
$$
\end{proposition}
\begin{proof}
From Proposition \ref{psiHrelation}, we have
${}_1\!\Psi_0\left[(a,A);-z\right]=\frac{1}{2 \pi i } \int_{C} \Gamma(-s)\Gamma(a+A s) z^{s} ds$. Let us now close the contour at infinity in the left half-plane. Then only the poles of $\Gamma(a+A s)$ will contribute:
$$
\Gamma(a+A s)\xrightarrow[s\rightarrow -\frac{n+a}{A}]{} \frac{(-1)^n}{n! A}\frac{1}{s+\frac{s+n}{A}}\ .
$$
Then we immediately get
$$
{}_1\!\Psi_0\left[(a,A);-z\right] = \frac{z^{-\frac{a}{A}}}{A}\sum_{n=0}^{\infty}\frac{(-1)^n}{n!}z^{-\frac{n}{A}}\Gamma\left(\frac{a+n}{A}\right)\ .
$$
The leading contribution for $|z|\rightarrow\infty$ comes from $n=0$, which completes the proof.
\end{proof}

\noindent This result can be easily generalized to the case of arbitrary psi-functions, ${}_p\Psi_q$.

One last result we would like to obtain is about derivative of the Fox-Wright psi-function.
\begin{proposition}\label{derivative}
\begin{eqnarray}
& &\frac{d}{d z}{}_p\Psi_q\left[(a_1,A_1),...,(a_p,A_p);(b_1,B_1),...,(b_q,B_q); z\right] = \nonumber\\ & & ={}_p\Psi_q\left[(a_1+A_1,A_1),...,(a_p+A_p,A_p);(b_1+B_1,B_1),...,(b_q+B_q,B_q); z\right]\nonumber
\end{eqnarray}
\end{proposition}
\begin{proof}
The proof is by the definition of ${}_p\Psi_q$:
\begin{eqnarray}
& &\frac{d}{d z}{}_p\Psi_q\left[(a_1,A_1),...,(a_p,A_p);(b_1,B_1),...,(b_q,B_q); z\right] = \nonumber\\
& &=\sum_{k=1}^\infty\frac{\prod_{n=1}^p \Gamma(a_n+A_n k)}{\prod_{n=1}^q \Gamma(b_n+B_n k)}\frac{z^{k-1}}{(k-1)!} = \sum_{k=0}^\infty\frac{\prod_{n=1}^p \Gamma(a_n+ A_n+A_n k)}{\prod_{n=1}^q \Gamma(b_n+B_n+B_n k)}\frac{z^k}{k!} = \nonumber \\
& & ={}_p\Psi_q\left[(a_1+A_1,A_1),...,(a_p+A_p,A_p);(b_1+B_1,B_1),...,(b_q+B_q,B_q); z\right]\nonumber
\end{eqnarray}
\end{proof}


\end{document}